\documentclass[%
nofootinbib,%
]{revtex4-1}

\usepackage{graphicx}
\usepackage{dcolumn}
\usepackage{bm}

\usepackage[utf8]{inputenc}
\usepackage[T1]{fontenc}

\usepackage{amsfonts,amsmath,amssymb}
\usepackage[mathscr]{eucal}
\usepackage{times,theorem,latexsym}
\usepackage{setspace}

\newenvironment{proof}
{ \par\noindent{$ \mathsf{Proof\!:}$}\\[5pt]}
{ \phantom{xxxxx}\hspace{\stretch{1}}$\Box$}

\newenvironment{remark*}[1][]
{ \par\noindent{\bf Remark{#1}:} }
{ \hspace{\stretch{1}} }

\newtheorem{theorem}{Theorem}
\newtheorem{corollary}[theorem]{Corollary}
\newtheorem{definition}[theorem]{Definition}
{\theorembodyfont{\rmfamily}\newtheorem{example}[theorem]{Example}}
\newtheorem{lemma}[theorem]{Lemma}
\newtheorem{proposition}[theorem]{Proposition}
{\theorembodyfont{\rmfamily}}

\newcommand{\mylabel}[1]{\label{#1}}  

\def\ds{\displaystyle}
\def\e{\varepsilon}
\def\/#1{{\rm({\bf #1})} }

\def\Ps{\mathcal{C}}
\def\Fs{{\mathfrak M}}

\def\F{\mathscr M}
\def\H{\mathcal G}
\def\G{\mathscr N}

\def\h{{\mathfrak{g}}}
\def\f{{\mathfrak{m}}}
\def\ff{$\f$}

\def\P{{P}}
\def\Q{{Q}}
\def\dottt{\scalebox{0.8}{\ldots}}
\def\Fo{\F_+}
\def\fo{\f_+}
\def\hh{\tilde{\h}}

\def\br#1{\hfill \\[#1 pt]}

\newcommand{\delete}[1]{}
\newcommand{\dsskip}[1][0.69]{
  \abovedisplayskip=#1\abovedisplayskip
  \belowdisplayskip=#1\belowdisplayskip
}%
\newcommand{\dsskipp}[1][0.7]{
  \setstretch{0.9}
  \abovedisplayskip=#1\abovedisplayskip
  \belowdisplayskip=#1\belowdisplayskip
}%

\def\bb{\mathbb}

\def\N{{\bb N}}				    
\def\R{{\bb R}}				    
\def\Z{{\bb Z}}				    

\def\nrN{N}                                                       
\def\nrM{M}                                                      
\def\cf{\mathfrak{n}}                                         
\def\cfu{\cf_\star}                                              
\def\ff{\mathfrak{f}}                                           
\def\ffu{\ff_\star}                                                
\def\efN{\mathscr{N}}                                        
\def\efNm{\efN_\star}                                        
\def\pN{\mathscr{N}_p}                                     
\def\mN{\mathscr{N}_+}                                    
\def\mn{\mathrm{n}_+}                                      
\def\efF{\mathscr{F}}                                         
\def\efFm{\efF_\star}                                         
\def\w{c}                                                              
\def\W{C}                                                             
\def\setW{\mathcal{C}}                                        
\def\v{b}                                                               
\def\V{B}                                                              
\def\wl{w}                                                             
\def\Wl{W}                                                            
\def\setWl{\mathcal{W}}                                       
\def\setP{\mathcal{P}}                                          

\def\Nmaps{{\mathfrak N}}                                   
\def\scrk{{\scriptscriptstyle (k)}}                            
\def\Hsp{\mathscr{H}}                                           
\def\ketpsi{{\mid \! \psi \,\rangle}}                          
\def\keti{{\mid \! i \,\rangle}}                                   
\def\ketj{{\mid \! j \,\rangle}}                                   
\def\ketchim{{\mid \! \chi_m \,\rangle}}                  
\def\union{\cup}                               

\def\p{c}\def\P{C}\def\Q{B} 

\begin{document}


\title[]{Effective Number Theory: Counting the Identities of a Quantum State}

\author{Ivan Horv\'ath} 
\email[]{ihorv2@g.uky.edu} 
\affiliation{University of Kentucky, Lexington, KY, USA\\  ORCiD: 0000-0001-8810-9737}

\author{Robert Mendris}
\email[]{rmendris@shawnee.edu}
\affiliation{Shawnee State University, Portsmouth, OH, USA\\  ORCiD: 0000-0003-3267-3552}

\vskip 0.5in

\date{Jul 11, 2018 ; published version Nov 10, 2020}

\begin{abstract}
Quantum physics frequently involves a need to count the states, subspaces, 
measurement outcomes, and other elements of quantum dynamics. 
However, with quantum mechanics assigning probabilities to such objects, 
it is often desirable to work with the notion of a ``total'' that takes into 
account their varied relevance.
For example, such an effective count of position states available to a lattice 
electron could characterize its localization properties. 
Similarly, the effective total of outcomes in the measurement step 
of a quantum computation relates to the efficiency of the quantum algorithm.
Despite a broad need for effective counting, a well-founded prescription has not been 
formulated. Instead, the assignments that do not respect the measure-like nature of 
the concept, such as versions of the participation number or exponentiated entropies, 
are used in some areas. 
Here, we develop the additive theory of effective number functions 
(ENFs), namely functions assigning consistent totals to collections of objects endowed 
with probability weights. Our analysis reveals the existence of a {\em minimal total}, 
realized by the unique ENF, which leads to effective counting with absolute meaning. 
Touching upon the nature of the measure, our results may find applications not only in 
quantum physics, but also in other quantitative sciences.

\end{abstract}

\maketitle



\section{\label{sec:1} Motivation and Overview} 

Among the distinctive features of quantum mechanics is that, in some regards, a quantum system 
acts as though it is simultaneously in multiple states of a given type. As an extreme example, 
a lattice Schr\"odinger particle in a momentum eigenstate is often said to reside at all positions, 
having an equal chance of being detected anywhere. However, how many of such 
``position identities'' are effectively present in a generic state that can assign 
an arbitrarily varied relevance (probability) to different locations? 

Variants of this counting problem appear in quantum physics quite often. One 
example arises in the context of Anderson localization~\cite{And58A} 
(see e.g.~\cite{Mar06A,Eve08A} for reviews). 
Indeed, the Fermi-level electron inside the band of extended states (Anderson conductor)
is thought of as effectively present in most of the available position states.
In contrast, such an electron inside the band of localized states (Anderson insulator) only 
resides in a drastically reduced subset of them. 
Hence, 
a well-founded effective counting of states could be used to quantitatively describe 
the transition between these regimes in a novel way.

Viewing an Anderson electron from a different perspective, the effective state counting 
could also be used to analyze the indeterminacy (quantum uncertainty) associated with 
the measurement of its position. In such an approach, uncertainty would be represented by
the effective number of position states the electron collapses into upon repeating 
the experiment: the smaller this effective number, the smaller the uncertainty. Note that 
the general treatment of quantum indeterminacy this way (with respect to an arbitrary basis) 
would be very different in nature from e.g. the classic spectral approach~\cite{Hei27A, Ken27A}.


The above types of physics analyses may usefully materialize if the generic question [Q] 
below can be suitably formalized and meaningfully answered. In particular, if $\ketpsi$ is a state 
from $\nrN$-dimensional Hilbert space and 
$\{ \, \keti \} \equiv \{ \, \keti  \mid i=1,2,\ldots,\nrN \,\}$ 
its orthonormal basis, it is desirable to ask:\footnote{Note that this canonical setting already covers 
many-body and field-theoretic systems whose quantum dynamics can be defined via lattice 
regularization. The extension to continuously labeled bases will be explicitly given in the context 
of application to quantum uncertainty~\cite{Hor18B}.}


\medskip
\noindent 
     {\bf [Q]} {\em How many states from $\{ \, \keti \}$ is the system 
     described by $\ketpsi$ effectively in?} 
\medskip

\noindent
A well-founded resolution of this ``quantum identity problem'' is not readily 
available.\footnote{This is rooted in the fact that quantifiers of the desired type 
cannot be expressed as quantum-mechanical expectation values in state $\ketpsi$.}
In this paper, we develop a theoretical framework (effective number theory)
that gives the rationale to the following answer:

\medskip

   \noindent
   {\bf [A]} {\em Let $P \!=\! (p_1, p_2,\ldots,p_\nrN ) \,,\,
   p_i = \,\mid \!\! \langle \,i \!\mid \! \psi \,\rangle \!\! \mid^2$, be 
   the probability
   vector assigned to quantum state $\ketpsi$ and basis $\{ \, \keti \}$, and~let 
   $\W \!\!=\!\! (\w_1,\w_2,\ldots,\w_\nrN) \,,\, \w_i \!=\! \nrN p_i$. The system described by 
   $\ketpsi$ is effectively in $\efNm[\, \ketpsi , \{ \, \keti \} \,] \!=\! \efNm[\W]$ 
   states from $\{ \, \keti \}$, where} 
   \begin{equation}
      \efNm[\W]  \,=\, \sum_{i=1}^\nrN \cfu(\w_i)   \quad,\quad
      \cfu(\w)  \; = \;   \min\, \{ \w, 1 \}    \quad  \text{for all} \quad c \in [0,\infty)
      \label{eq:015}         
   \end{equation}
 
 
\noindent   
To arrive at [A], we start with the axiomatic definition of the {\em effective number 
function} (ENF) $\efN[\, \ketpsi , \{ \, \keti \} \,] \!=\! \efN[\W]$, namely 
a function consistently assigning the effective totals. Solving [Q] then amounts to 
finding such an $\efN$ and using it to specify the effective number of quantum 
identities in all situations. 
The subsequent analysis shows, however, that there exists an entire continuum
of ENFs. This could render each fixed choice of $\efN$ too arbitrary 
and its individual value uninformative on its own.\footnote{ENFs would still be 
useful since, by construction, each of them individually conveys a universal 
comparative information about effective totals.} 
Interestingly, this is not the case because it turns out that $\efNm$
is an ENF with absolute meaning. Indeed, we will prove that 
$\efNm[\W] \le \efN[\W]$ for all $\W$ and 
all~$\efN$, making $\efNm$ the unique minimum (least element) on the set of 
all ENFs. Having revealed that the system in state $\ketpsi$ has to be 
characterized as being simultaneously in at least 
$\efNm[\, \ketpsi , \{ \, \keti \} \,]$ states from $\{ \, \keti \}$, 
this result is used in [A] as a basis for the meaningful canonical 
choice of ENF. It should be noted in this regard that a maximal ENF, 
whose interpretation would otherwise be on equal footing with $\efNm$, 
does not exist (see Theorem~2).\footnote{
The existence of multiple ENFs endows the constructed framework with 
flexibility to accommodate quantum identity problems more structured 
than [Q]. This may entail an additional problem-specific constraint(s) 
on an ENF, possibly leading to a unique or privileged choice other than 
$\efNm$. However, a generic extra requirement is that the effective total
determines the subset of $\{ \, \keti \}$ in which $\ketpsi$ is effectively present.
For example, in the context of Anderson localization, it is of interest
to identify the spatial region effectively occupied by the electron.
It can be shown that $\efNm$ is the only ENF leading to a consistent selection 
of such {\em effective support} of $\ketpsi$ on $\{ \, \keti \}$.}

%
 
A crucial novelty in our approach is the inclusion of additivity as a requirement 
for ENFs. This step is necessary since the effective number of states is an additive 
concept. However, a proper formulation requires some care. To that end, as well as 
to start invoking parallels with localization, consider the simple setting 
of a spinless Schr\"odinger particle on a finite lattice. 
In the position basis, its state $\ketpsi$ is represented by $\nrN$-tuple 
$( \psi(x_1),\ldots,\psi(x_\nrN) )$, with $p_i = \psi^\star \psi(x_i)$ being 
the probability of detection at the location $x_i$. Denoting by $\setW$ the set of
all {\em counting} weight vectors $\W=(\w_1,\ldots,\w_\nrN)$, $\w_i \!=\!\nrN p_i$, 
namely\footnote{Working with counting vectors \eqref{eq:005} rather than probability 
vectors $P \in \setP  =  \mathop{\cup}_{\nrN}  \, \setP_\nrN$ is simply a matter 
of convenience. All results translate straightforwardly.} 
\begin{equation}
    \setW=\union_\nrN \setW_\nrN    \qquad , \qquad
    \setW_\nrN  \,=\, 
    \bigl\{ \, (\w_1,\w_2, \ldots, \w_\nrN) \;\mid\;  \w_i \ge 0 \,,\,\, 
    \sum_{i=1}^\nrN \w_i = \nrN \, \bigr\}  \quad
    \label{eq:005}
\end{equation}
the additivity property for $\efN$ arises as follows. Assume that the particle is 
restricted to a lattice of $\nrN_1$ sites in a state generating the weight vector 
$\W_1 \!\in\! \setW_{\nrN_1}$. Separately, let it be restricted to a non-overlapping 
adjacent lattice of $\nrN_2$ sites and characterized by $\W_2 \!\in\! \setW_{\nrN_2}$.  
With symbol $\boxplus$ representing the concatenation 
operation\footnote{If $\W \!=\! (\w_1,\ldots, \w_\nrN) \in \setW_\nrN$ and 
$\V \!=\! (\v_1,\ldots,\v_\nrM) \in \setW_\nrM$, then 
$\W \boxplus \V \equiv (\w_1, \ldots, \w_\nrN, \v_1, \ldots, \v_\nrM)$.}, 
since
\begin{equation}
      \W \,=\, \W_1  \boxplus \W_2 \,\in\, \setW_{\nrN=\nrN_1+\nrN_2}  
       \label{eq:025}         
\end{equation}
there exists a state of the particle on the combined lattice, producing this composite 
$\W$. Given the additivity of numbers, the sum rule for the number of available states 
($\nrN = \nrN_1 + \nrN_2$) 
has to hold for its effective counterpart as well 
(\,$\efN[\W] = \efN[\W_1] + \efN[\W_2]$\,).

Consequently, the additivity property that we impose is
\begin{equation}
     \efN \bigl[ \W_1 \boxplus \W_2, \, \nrN_1 + \nrN_2 \bigr]  \;=\; 
     \efN \bigl[ \W_1, \nrN_1 \bigr]  \,+\, \efN \bigl[ \W_2, \nrN_2 \bigr]   
     \quad , \quad \forall \, \W_1, \W_2   \quad
     \tag{A}
     \label{eq:add}         
\end{equation}
Here, the dimensions of vector arguments were made explicit to emphasize that 
$\efN[\W,\nrN]$ represents $\nrN$ modified by distribution $\W$. Notice that $\efNm$ is 
evidently additive and that the above reasoning does not depend on the system, state, 
or basis in question. 

Several decades ago, Bell and Dean~\cite{Bel70A} dealt with a problem analogous to [Q] 
while analyzing the localization properties of vibrations in glassy silica. In particular, 
they asked how many atoms do these vibrations effectively spread over. 
Their quantifier, the participation number $\pN$, is given by 
\begin{equation}
    \frac{1}{\pN[\W]} \,=\, \frac{1}{\nrN^2} \sum_{i=1}^\nrN \w_i^2  \quad
    \label{eq:045}              
\end{equation}
and is still widely used in the analysis of localization. In other areas, it is common 
to exponentiate a suitable entropy, such as the Shannon~\cite{Sha48A} or 
R\'enyi entropies~\cite{Ren60A}, and use it
for analogous purposes. However, none of these quantifiers is \eqref{eq:add}-additive.
Their interpretation as effective totals is thus vague and they tend to be too 
arbitrary. In contrast, incorporating additivity into the definition of ENFs leads to 
the resolution of the quantum identity problem and suggests new possibilities 
both in physics and measure-related aspects of 
mathematics. The effective number theory, which we develop here as a tool to solve [Q], 
provides a theoretical starting point for such developments. 

In the rest of this section, we describe the construction of ENFs and discuss the key 
results of effective number theory. The goal here is to provide a concise but rigorous 
overview, including the motivations for axiomatic properties, as well as the ramifications 
of deduced features. A fully mathematical treatment in the technically convenient dual 
form of effective complementary numbers (co-numbers) is then given in 
Sec.~\ref{sec:eff_numbers}. Various generalizations of the quantum identity problem 
are discussed in Sec.~\ref{sec:count_idents}.
We then outline the use of effective numbers in quantum theory from 
a very broad perspective, namely as a general tool to characterize quantum states 
(Sec.~\ref{sec:str_states}). 
Concluding remarks are given in Sec.~\ref{sec:conclude}.


\subsection{Effective Numbers}
\label{ssec:effnums_def}

We now develop the notion of ENF as a function $\efN \!=\! \efN[\W]$, assigning an
effective total to each distribution of weights $\W \in \setW$ over the elements of a basis. 
Such a construction clearly does not depend on the fact that counted objects are quantum 
states, and we will thus use generic terms in that regard from now on. The underlying goal 
is to extend the ``counting measure'' for a collection of distinct, but otherwise equivalent 
objects (natural number $\nrN \in \N$) to the situation when these objects acquire varied 
importance expressed by their counting weights (effective number $\efN[\W] \!\in \R$). 
The additivity property \eqref{eq:add} is thus a basic consistency requirement for 
acceptable ENFs.

Like in ordinary counting, no specific relation among individual objects is assumed. 
Thus, in the same way the number of balls in a bag does not change upon their reshuffle, 
the effective number will not change upon the permutation of counting weights. In other 
words, ENFs are required to be totally symmetric in their 
arguments, namely\footnote{We write $\efN(\w_1,\ldots,\w_\nrN)$ when weights need to be 
distinguished, but use the functional notation $\efN[\W]$ otherwise.}
\begin{equation}
    \efN(\ldots \w_i  \ldots \w_j  \ldots)  \,=\,  \efN(\ldots \w_j  \ldots  \w_i  \ldots)
    \quad\; , \quad\; \forall \, i \ne j
     \quad
     \tag{S}
     \label{eq:sym}             
\end{equation}    

Extensions $\nrN \rightarrow \efN[\W]$ are by definition such that ordinary counting corresponds 
to all objects being equally important, and thus to a uniform distribution. More precisely, 
\begin{equation}
     \efN(1,1,\ldots,1) = \nrN      \quad\; , \quad\;  
     (1,1,\ldots,1) \in \setW_\nrN  \quad\; , \quad\;
     \forall \, \nrN
     \quad
     \tag{B1}
     \label{eq:bou1}                   
\end{equation}
On the other hand, whenever all the weight is given to a single object, all others being irrelevant, 
the effective number is required to be one, namely 
\begin{equation}
     \efN(\nrN,0,\ldots,0) \,=\, \efN(0,\nrN,0,\ldots,0) \,=\, \ldots \,=\, 
     \efN(0,\ldots,0,\nrN) \,=\, 1 
     \quad
     \tag{B2}
     \label{eq:bou2}                   
\end{equation}
within each $\setW_\nrN$.
Note that $(1,1,\ldots,1) \in \setW_\nrN$ and $(\ldots,0,\nrN,0,\ldots) \in \setW_\nrN$ 
are the opposite extremes in the cumulation of weight. Hence, the effective number of objects 
with arbitrary weights has to fall between the corresponding extremal values, namely
\begin{equation}
     1 \, \le \, \efN[\W] \, \le \, \nrN   \quad , \quad \forall \, \W \in \setW_\nrN 
                                                      \quad , \quad   \forall \, \nrN   \quad    
     \tag{B} 
     \label{eq:bou}                                                                 
\end{equation}

The degree of weight cumulation plays a more detailed role in effective numbers than 
just determining the boundary properties. Indeed, the concept has to respect that
increasing the cumulation in the distribution cannot increase the effective number. 
To formulate such monotonicity, consider two objects 
weighted by $\W \!=\! (\w_1,\w_2) \in \setW_2$ with $\w_1 \le \w_2$. The deformation 
$\W \rightarrow \W_\epsilon = (\w_1-\epsilon, \w_2+\epsilon)$ leads to further cumulation 
in favor of the second object, and thus, $\efN[\W_\epsilon] \le \efN[\W]$ is imposed for all
$0 \le \epsilon \le \w_1$. In a situation with an arbitrary $\nrN$, we require the same for each 
ordered pair $\w_i \le \w_j$ and deformation $0 \le \epsilon \le \w_i$, 
namely\footnote{To visualize how the elementary deformation in \eqref{eq:mon-} increases 
cumulation, one may picture each object as a cylindrical column of incompressible liquid in 
the amount of its counting weight. Arranging the columns by increasing height from the left 
to the right produces a half-peak profile with cumulation on the right. Consider the segment of 
this profile delimited by columns $\w_i$ and $\w_j$ entering \eqref{eq:mon-}. The monotonicity 
operation is represented by the transverse flow of liquid from the left to the right endpoint 
through columns between them. It is understood that the columns are ordered at every moment of 
the flow and thus, as the amount of liquid at the endpoints changes, the length of 
the segment may increase. Since the liquid flows toward the center of cumulation at every 
point of the process, the resulting distribution is more cumulated than the original one.}
\begin{equation}
      \efN(\ldots \w_i - \epsilon  \ldots  \w_j + \epsilon  \ldots)   \,\le\, 
      \efN(\ldots  \w_i  \ldots  \w_j  \ldots)   
      \quad   
      \tag{M$^-$}
      \label{eq:mon-}                   
\end{equation}
It is easy to check that \eqref{eq:mon-}-monotonic $\efN$ attains its maximal value over 
$\setW_\nrN$ at $(1,1,\ldots,1)$, while the minimum is at one or multiple fully cumulated 
vectors $(\ldots,\nrN,\ldots)$. Conditions \eqref{eq:bou1}, \eqref{eq:bou2}, and \eqref{eq:bou} 
are thus compatible with~\eqref{eq:mon-}.\footnote{\label{note:schur} Monotonicity 
\eqref{eq:mon-} is closely related to Schur concavity. The latter is equivalent 
to imposing \eqref{eq:mon-} and symmetry \eqref{eq:sym} simultaneously 
(see e.g.~\cite{Arn87A}).}
Note that, although not an ENF, the participation number \eqref{eq:045} satisfies 
\eqref{eq:mon-} monotonicity.

\begin{figure}[t!]
\begin{center}
    \centerline{
    \hskip 0.00in
    \includegraphics[width=0.75\columnwidth,angle=0]{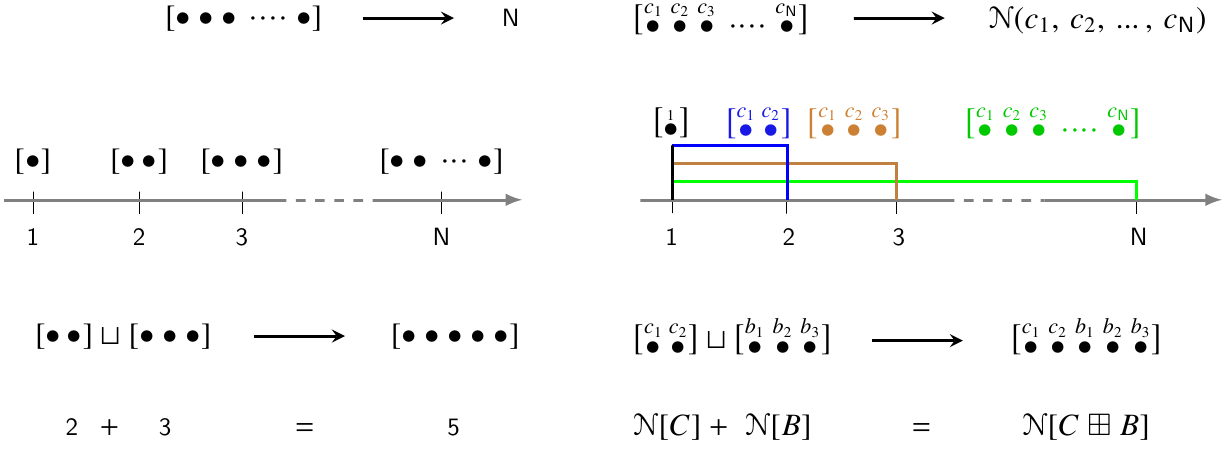}
     }
     \vskip -0.06in
     {\caption{Schematic representation of the extension from natural numbers (left) to effective numbers 
     (right). See the discussion in the text.}
     \label{fig:numbers}}   
     \vskip -0.45in
\end{center}
\end{figure}

The final requirement in the definition of ENFs is continuity. The nature of problems with admitting 
discontinuities can be illustrated by 
\begin{equation}
     \mN[\W] = \sum_{i=1}^\nrN \mn(\w_i)    \quad , \quad
     \mn(\w)  \,=\,
     \begin{cases} 
     \;0 \;,   &   \;  \w=0 \\[8pt]
     \;1 \;,   &   \;  \w > 0
     \end{cases} 
     \quad
     \label{eq:075}         
\end{equation}
which counts the number of non-zero weights in $\W$ and will be relevant later in our analysis. 
Consider again two objects with $\W \!=\! (\w,2-\w)$. When $\w$ approaches zero, thus marginalizing 
the first object to an arbitrary degree, the effective number should approach one. However, this does 
not materialize in $\mN$ due to its discontinuity. In general, we require that the ENF cannot jump 
upon an arbitrarily small change of weights, namely
\begin{equation}
    \efN = \efN[\W]  \;\,
    \text{is continuous on} \;\, \setW_\nrN 
    \quad , \quad \forall \,\nrN  \quad 
    \tag{C}
    \label{eq:con}             
\end{equation}   
 
The properties discussed above define the set $\Nmaps$ of all effective number functions. 
However, there are dependencies among these requirements. In particular, it can be easily 
checked that the boundary condition \eqref{eq:bou1} is a consequence of \eqref{eq:bou2} 
and additivity. Similarly,  \eqref{eq:bou} follows from \eqref{eq:bou1}, \eqref{eq:bou2}, 
symmetry and monotonicity. This leaves us with

\medskip

\noindent {\bf Definition~0.}   {\em A real-valued function $\efN \!=\! \efN[\W]$ on $\setW$ is  
called an effective number function (belongs to set $\Nmaps$) if it is simultaneously additive 
\eqref{eq:add}, symmetric \eqref{eq:sym}, continuous \eqref{eq:con}, monotonic \eqref{eq:mon-} 
and satisfies the boundary condition \eqref{eq:bou2}.}

\medskip

Some of the features imprinted on the corresponding notion of effective numbers are 
visualized in Fig.~\ref{fig:numbers}.\footnote{Once an ENF is fixed and used to assign totals, 
the conventional ``number of objects'' is replaced by the ``effective number of objects''. 
While we use the term effective number only in this restricted sense here, the underlying
algebraic structure makes the concept similar to standard types of numbers.} 
On the left, natural numbers are shown as a theoretical model for expressing and 
manipulating the quantities of like objects (bags of balls) or of varied objects treated 
as equivalent. The bags containing differing amounts are assigned different discrete points 
on the real axis (natural numbers), with the operation of ``merging the bags'' (\,$\sqcup$\,) 
realized by ordinary addition. Extension to objects distinguished by counting weights is 
shown on the right. Here, the bags assigned equal amounts $\nrN$ by ordinary counting
may be assigned different effective numbers $\efN$, depending on the cumulation of their 
weight distributions. With maximal cumulation ($\delta$-function) 
producing $\efN \!=\! 1$, the effective number continuously and monotonically 
increases as cumulation decreases, reaching $\efN \!=\! \nrN$ when cumulation is absent 
(uniform distribution). The operation of merging bags is represented by the additivity 
property \eqref{eq:add}. Each element of $\Nmaps$, if any, implements a specific
version of this scheme. 
Thus, to assess the conceptual value and practical impact of effective 
numbers, it is necessary to decipher the structure of $\Nmaps$.


\subsection{Effective Counting}
\label{ssec:eff_counting}

It is not difficult to establish that ENFs do exist. For example, one can verify that 
the one parameter family of functions 
\begin{equation}
     \efN_{(\alpha)}[\W] = \sum_{i=1}^\nrN \cf_{(\alpha)}(\w_i)    \quad , \quad
     \cf_{(\alpha)}(\w)  = \min\, \{ \w^\alpha, 1 \}               \quad , \quad
     0 < \alpha \le 1
     \quad
     \label{eq:085}         
\end{equation}
belongs to $\Nmaps$, with $\efN_{(1)} \!=\! \efNm$. However, it is rather remarkable 
that all $\efN \in \Nmaps$ have the additively separable structure of \eqref{eq:085}. 
Indeed, Theorem~\ref{th:eff} (Sec.~\ref{sec:eff_numbers}) implies 
the following central result specifying $\Nmaps$ explicitly.

\bigskip

\noindent 
{\bf Theorem 1.}   {\em Function $\efN$ on $\setW$ belongs to $\Nmaps$ 
if and only if there exists a real-valued function $\cf=\cf(\w)$ on $[0,\infty)$ 
that is concave, continuous, $\,\cf(0)=0\,$, $\,\cf(\w) = 1\,$ 
for $\,\w \ge 1$, and}
\begin{equation}
     \efN[\W]  \,=\,  \sum_{i=1}^\nrN \cf(\w_i)   \quad , \quad
     \forall \,\W \in \setW_\nrN    \quad , \quad
     \forall \,\nrN \quad
     \label{eq:095}         
\end{equation}
{\em Such a function $\cf$ associated with $\efN \in \Nmaps$ is unique.}

\bigskip

\noindent
Thus, there is one-to-one correspondence between ENFs and functions of the 
single variable specified by Theorem~1.\footnote{Note that it suffices to require 
continuity at $\w=0$ since concavity guarantees it elsewhere.} Such $\cf$ 
associated with the given $\efN$ will be referred to as its {\em counting function}.


The necessity of the additively separable form \eqref{eq:095} for ENFs is interesting
conceptually. Indeed, it is common and familiar to represent the ordinary total 
(natural number) by a sequential process of adding a unit amount for each object 
in the collection. According to Theorem~1, this applies to every consistent extension 
to the effective total (effective number), albeit with objects contributing 
weight-dependent amounts specified by the counting function. It thus turns out
that the construction of ENFs generalizes the process of ordinary counting 
to the process of {\em effective counting}.


\subsection{Minimal Effective Number}

A key insight into the nature of effective counting is provided by the following 
results concerning the structure of set $\Nmaps$. They follow directly from 
Theorem~\ref{th:max} in Sec.~\ref{sec:eff_numbers}.  
\bigskip

\noindent {\bf Theorem 2.}   {\em Let $\efNm \!\in\Nmaps$ and $\mN \!\notin\! \Nmaps$ be 
functions on $\setW$ defined by \eqref{eq:015} and \eqref{eq:075}, respectively. Then} 
\begin{displaymath}     
     \begin{aligned}
         (a) \;\, & \;
               \efNm[\W]   \,\le\,  \efN[\W]  \,\le\,  \mN[\W]  \quad\; , \quad\;  
               \forall\, \efN \in \Nmaps \;\;,\;\;  \forall\, \W \in \setW     \\[4pt]
         (b) \;\, & \;      
               \bigl\{ \, \efN[\W]  \,\mid\, \efN \in \Nmaps \, \bigr\}  \;=\;  
               [\, \alpha , \beta \,]   \quad\; , \quad\;
               \alpha=\efNm[\W] \;\;,\;\; \beta=\mN[\W]
               \;\;,\;\;  \forall\, \W \in \setW
      \end{aligned}       
\end{displaymath}     


To elaborate, first note that $(a)$ is the refinement of defining condition \eqref{eq:bou}. 
While the upper bound is intuitive ($\,\mN[\W]$ counts the number of non-zero weights in $\W\,$),
the lower one is unexpected and consequential. In particular, the effective number of objects
weighted by $\W$ cannot be smaller than $\efNm[\W]$. Since $\efNm$ is an ENF, this feature 
is inherent to the concept itself: there is a meaningful notion of the minimal effective number.
In technical terms, $\efNm$ is the least element of function set $\Nmaps$ with respect 
to partial order 
$(\,\efN_1 \le \efN_2\,) \,\Leftrightarrow\, (\,\efN_1[\W] \le \efN_2[\W], \, \forall \, 
\W\in \setW \,)$, and thus a unique ENF with this property.

Part $(b)$ conveys that, for each fixed $\W \!\in\! \setW$, effective counting can be 
adjusted so that $\efN[\W]$ assumes any desired value from the allowed range specified 
by $(a)$. While reflecting a certain degree of arbitrariness built into the concept of effective 
numbers, the associated freedom of choice is in fact quite natural. To illustrate this, 
consider $\nrN$ objects with non-zero weights of very disparate magnitudes so that the collection 
is usefully characterized by an effective number. The insistence on ordinary count in this 
situation constitutes a ``large extrapolation'' since it forces each object to contribute 
equally despite the disparity. According to $(b)$, such extrapolation can be realized by 
a sequence of ENFs that bring the effective total arbitrarily close to $\nrN$. 
Accommodating the needed continuum of consistent schemes can thus be considered a 
useful feature in a framework describing the generalized aspects of counting. 


Note that $(b)$ also confirms an intuitive expectation that there is no maximal ENF since, 
although specifying a supremal value for each $\W$, function $\mN$ does not belong to 
$\Nmaps$. 
Taken together, the results of Theorem~2 form the basis for our canonical solution [A] of 
the quantum identity problem [Q]. The existence of minimal total $\efNm$ is particularly 
consequential in applications of effective numbers. One notable example is that it facilitates 
the notion of minimal (intrinsic) quantum uncertainty~\cite{Hor18B}. 


\section{Effective Number Theory}
\label{sec:eff_numbers}


In this section, we will develop the theory of effective numbers with the requisite mathematical 
detail. The aim is to do this in a self-contained accessible manner using elementary mathematics. 
Certain generalizations regarding the underlying algebraic structure will be elaborated upon in 
a separate mathematical account.


\subsection{Effective Complementary Numbers}


It turns out that there are several practical advantages to carrying out this discussion in terms 
of effective complementary numbers (effective co-numbers) realized by 
functions\footnote{They could also be called the effective dual numbers but 
the duality is not a central property here.}
\begin{equation}
   \F[\P] = N-\G[\P] \quad,\quad  \P\in\Ps_N  \quad,\quad  \G \in \Nmaps
\end{equation}
where $\G=\G[\P]$ are the ENFs introduced in Sec.~\ref{ssec:effnums_def}. Following this
route, we start by the explicit definition of effective co-number functions (co-ENFs) entailed by 
the above relationship.

\begin{definition}\mylabel{def:fun}\setstretch{1.2}
  $\Fs$ is the set of effective co-number functions $\F$, where $\F :\, \Ps \rightarrow \R\,$ have 
  the following properties:
 \phantom{}for all $N,M \in \Z^+$, for all integer $1\le i,j \le N$, $i\ne j$, for all 
 $\P = (\p_1,...,\p_N) \in \Ps_N$, and for all $\Q \in \Ps_M$,

\medskip\noindent
   \/A additivity: $\F[\P \boxplus \Q] = \F[\P] + \F[\Q]$ 
  \\ \/{co-B2} boundary values: $\F(N,0,...,0) = N-1$, where  $(N,0,...,0)\in \Ps_N$ 
  \\ \/C continuity of $\F$ restricted to $ \Ps_N$
    whose topology is inherited from the standard topology on $\R^N$ 
  \\ \/{M$^+$} monotonicity: $0 < \e \le \min\{\p_i,N-\p_j\}$,\ \ $\p_i \le \p_j  \;\Rightarrow\;$
                $\F(...,\p_i,...,\p_j,...) \le \F(...,\p_i-\e,...,\p_j+\e,...)$ 
  \\ \/S symmetry: $\F(...,\p_i,...,\p_j,...) = \F(...,\p_j,...,\p_i,...)$
\end{definition}

\medskip
The following examples will be useful in the course of our analysis.

\begin{example}\mylabel{ex:f_a}
The function $\F_{(\alpha)}[\P] = \sum_i \f_{(\alpha)}(\p_i) = \sum_{\p_i=0} 1 + \sum_{\p_i \in (0,1)} (1 - \p_i^\alpha)$, where
\dsskipp
\begin{align*}
  \f_{(\alpha)}(\p) =
  \begin{cases}
    1,  &  \p = 0
    \\
    1-\p^\alpha,  & 0 < \p < 1 
    \\
    0,  &      1 \le \p
  \end{cases}
\end{align*}
belongs to $\Fs$ for $\alpha \in (0,1]$. This example is complementary to $\G_{(\alpha)}$ introduced in (6).
\end{example}

\begin{example}\mylabel{ex:fmin}
The function $\Fo[\P]=\F_{(0)}[\P]=\sum_{i=1}^N \fo(\p_i)$, where
\dsskipp
\begin{align*}
  \fo(\p) = \f_{(0)}(\p) =
  \begin{cases}
    1,  &  \p=0 
    \\
    0,  &  0 < \p
  \end{cases}
\end{align*}
satisfies Definition \ref{def:fun} except for continuity. Thus, $\Fo \notin \Fs$
is complementary to $\G_+$ in (5).
\end{example}

\begin{example}\mylabel{ex:flin}
The function $\F_\star=\F_{(1)}[\P]=\sum_{i=1}^N \f_\star(\p_i)$, where
\dsskipp
\begin{align*}
  \f_\star(\p) = \f_{(1)}(\p) = 
  \begin{cases}
    1-\p,  &  0 \le \p < 1 
    \\
    0,  &  1 \le \p
  \end{cases}
\end{align*}
satisfies all the properties from Definition \ref{def:fun}, so $\F_\star \in \Fs$.
This co-ENF is complementary to $\G_\star$ in (1).
\end{example}


Due to its repeated use in what follows, it is useful to formalize the following obvious Lemma.

\begin{lemma}\label{le:p1p}\setstretch{1.5}%
  If $\F$ satisfies (A) and (S) then  for all $N$
  \\ \/i $\F[\P] = \F(\p_1,\dottt, \p_{i-1}, 1, \p_{i+1},\dottt, \p_{N}) = \F(\p_1,\dottt, \p_{i-1}, \p_{i+1},\dottt, \p_{N}) + \F(1)$,
  where $\P\in \Ps_N$,
  \\ \/{ii} $\F(1,\dottt, 1) = N\F(1)$, where  $(1,\dottt, 1)\in \Ps_N$.
\end{lemma}


\subsection{Separability}

\noindent We now focus on demonstrating the results that will ultimately clarify the content of set $\Fs$ 
and thus of set $\Nmaps$. The main conclusion of the analysis here is that all co-ENFs 
are of an {\em additively separable} form, such as the one exhibited by the family of functions 
$\F_{(\alpha)}$. 
The property of additive separability is defined as follows.\footnote{Unless stated otherwise, 
referencing ``function'' {\em in this section} applies to both real and complex-valued function, and 
referencing ``number'' applies to both real and complex options. }


\begin{definition}\mylabel{df:separable} 
    Additively separable function $\H$ on $\Ps$ is one that can be expressed as 
    \dsskip
    \begin{align}\mylabel{eq:separable}
         \H[\P] = \sum_{i=1}^N \h(\p_i)  \quad , \quad \P\in \Ps_N \quad ,\quad N=1,2,\ldots 
    \end{align}
    where $\h(\p)$ is some function defined on $[0,\infty)$. Function $\h(\p)$ is called a generating function 
    of $\H[\P]$.
\end{definition}

\medskip\noindent
Additively separable $\H$ is generated by infinitely many distinct functions.
However, for co-ENFs, a canonical representative can be singled out that is continuous and 
bounded on $[0,\infty)$ (see Proposition \ref{prop:generating} and Corollary \ref{cor:bounded}).

\smallskip

The relevant insight into additive separability is provided by Lemma~\ref{le:separable} below. 
Before formulating it, let us associate with every $\P \!\in\! \Ps_N$ 
the vector $\P^\uparrow \!\in\! \Ps_N$ obtained by permuting the components of $\P$ into ascending 
order. Furthermore, $\P^\uparrow_<$ will denote a vector obtained from $\P^\uparrow$ by keeping only 
components less than one and removing the rest. Note that for any symmetric function $\F$ on $\Ps$, 
we have $\F[\P]=\F[\P^\uparrow]$ and also that $\P^\uparrow_< \notin \Ps$.  
In Proposition~\ref{prop:monoseparable} we will work with $\P^\uparrow_\le$ with an
analogous meaning. 

\medskip

\begin{lemma}{\bf (Separability)}\label{le:separable} 
  Let $\H$ be a function on $\Ps$ satisfying (A), (S) and the property
  \dsskip
  \begin{align}\label{eq:(<)}
     \forall \,N  \;\;,\;\; \forall\, \P \,,\Q \in \Ps_N \qquad:\qquad   \P^\uparrow_< \,=\, \Q^\uparrow_< 
         \quad\Longrightarrow\quad \H[\P] = \H[\Q]
  \end{align}
  Then, the following statements hold:
  \br7
     \/a $\H$ is additively separable.
  \br7
     \/b  If, in addition, $\H$ is continuous\footnote{Continuity on $\Ps_2$ may appear weaker than $(C)$ continuity, but this lemma shows that they are equivalent.}
     on $\Ps_2$, then there exists a continuous function generating it.
\end{lemma}


\begin{proof} 
  \setstretch{1.2}%
  \/{a} Assuming $\P \ne (1,1,\ldots,1)$, let $\P^\uparrow=(\p_1, \dottt, \p_N) \in\Ps_N$ and
     $\P^\uparrow_< = (\p_1, \dottt, \p_m)$. We will distinguish two cases, namely $2m \le N$ and $2m > N$,
     for which we respectively get by using \eqref{eq:(<)}  
  \begin{eqnarray*}
     \H[\P] = \H(\p_1,\dottt, \p_m, \p_{m+1},\dottt, \p_N)
      = \H(\p_1,\dottt, \p_m, 2-\p_1, \dottt, 2-\p_m, 1,\dottt,1),
     \\ \H[\P\boxplus (1,\dottt,1)]  = \H(\p_1,\dottt, \p_m, \p_{m+1},\dottt, \p_N, 1,\dottt,1)
     = \H(\p_1,\dottt, \p_m, 2-\p_1, \dottt, 2-\p_m),
   \end{eqnarray*}
   with $2-\p_\ell > 1$ for $\ell=1,...,m$. The vectors on the top line (case $2m \le N$) are from $\Ps_{N}$,
   while those on the bottom line (case $2m > N$) are from $\Ps_{2m}$.
   In both cases, Lemma \ref{le:p1p}, symmetry (S), and additivity (A) lead to
   \begin{eqnarray*}
       \H[\P] &=& \H(\p_1, 2-\p_1,\dottt,\p_m, 2-\p_m) \,+\, (N-2m) \, \H(1)
       \\   &=& \H\bigl[ (\p_1, 2-\p_1) \boxplus \dottt \boxplus (\p_m, 2-\p_m) \bigr]  \,+\, (N-2m) \, \H(1)
       \\   &=& \sum_{\ell=1}^m \bigl( \, \vphantom{F^f} \H(\p_\ell,2-\p_\ell) - \H(1) \,\bigr) \,+\, (N-m) \, \H(1).
   \end{eqnarray*}
  Consequently, introducing the generating function
  \dsskipp
  \begin{align}\label{eq:f1}
    \h(x) = 
    \begin{cases}
	\H(x,2-x) - \H(1),  & x\in [0,1]
	\\
        \H(1),  & x \in (1,\infty)
    \end{cases}    
  \end{align}
  facilitates the claimed separability $\H[\P] = \sum_{i=1}^N \h(\p_i)$. 
  Note that for $\P = (1,1,\ldots,1)$, which was initially excluded, 
  the separability holds in the same form. 
\br9
\setstretch{1.1}%
  \/b Given the proof of (a), it is sufficient to show that the continuity of $\H$ on $\Ps_2$ implies 
  the continuity of $\h$ in \eqref{eq:f1}.  For that, one only needs to ascertain the continuity 
  at the gluing point $x=1$, which holds since we have two continuous functions with the same value 
  at the gluing point: $\H(1,2-1) - \H(1) = \H(1)$. 
\end{proof}

\bigskip
We will now demonstrate that all co-ENFs satisfy \eqref{eq:(<)}, and hence, they are additively 
separable.

\begin{proposition}\label{prop:monoseparable}
  All functions $\F \in \Fs$ are additively separable.
\end{proposition}
\begin{proof} 
  Because of symmetry (S), we will without loss of generality assume $\P=\P^\uparrow$, i.e. $\P$ is in ascending order.
  For $\P = (1,1,\ldots,1)$ the implication in \eqref{eq:(<)} is vacuously true\footnote{$\P^\uparrow_<$ is undefined here.},
  so $\P \ne (1,1,\ldots,1)$ is assumed in what follows.
  We will use index $\ell$ to label the elements of $\P^\uparrow_\le$ and index $j$ for the rest
  of the entries in $\P$. Hence,
  $\p_\ell \le 1 < \p_j$ for $\ell=1,...,m$, $\,j=m+1,...,m+n=N$, and $n=\sum_j 1$. 
    Then, by monotonicity (M$^+$), Lemma \ref{le:p1p}, and $\F(1)\!=\!0$, which is the $N\!=\!1$ case of (co-B2), we have:
  \begin{eqnarray*}
     \F(\dottt \p_\ell \dottt,\; \dottt \p_j \dottt)
     &\le \F(\dottt \p_\ell \dottt,\; \p_{m+1}-\e, \p_{m+2}, \dottt,\hspace{8mm} &   \p_{m+n-1}, \p_{m+n}+\e)\;\;,\;\,
       \mathrm{\ where\ } \e=\p_{m+1}-1
     \\ &=\F(\dottt \p_\ell \dottt,\; 1,\phantom{xxxxxi} \p_{m+2}, \dottt,\hspace{8mm}&  \p_{m+n-1}, 1+(\p_{m+n}-1)+(\p_{m+1}-1))
     \\ &\le \F(\dottt \p_\ell \dottt,\; 1,\phantom{xx} \p_{m+2}-\e, \p_{m+3}, \dottt,&  \p_{m+n-1}, 1+(\p_{m+n}-1)+(\p_{m+1}-1)+\e)
     \\ &\le \F\left( \dottt \p_\ell \dottt,\; 1,\dottt,1,\; 1 + \sum (\p_j-1) \right)&
     \\ &= \F\left( \dottt \p_\ell \dottt,\;\phantom{xxxxxii} 1-n + \sum \p_j \right).&
  \end{eqnarray*}
  The opposite inequality follows from additivity (A), (co-B2), and (M$^+$).
  With $\lceil x \rceil$ denoting the ceiling function, we start with $\sum (\lceil \p_j\rceil - 1)$ zeroes:
  \begin{eqnarray*}
     \F(\dottt \p_\ell \dottt,\; \dottt \p_j \dottt)
     = \F\left( \dottt \p_\ell \dottt,\; \dottt \p_j \dottt,\; 0,\dottt, 0,\; 1+\sum_{j=m+1}^{m+n} (\lceil \p_j\rceil - 1) \right)
       - \sum_{j=m+1}^{m+n} (\lceil \p_j\rceil - 1)
     \\ \ge \F\left( \dottt \p_\ell \dottt,\; \p_{m+1}+\e, \p_{m+2}, \dottt, \p_{m+n},\;  0,\dottt, 0,\; 1+\sum (\lceil \p_j\rceil - 1) -\e \right)
       - \sum (\lceil \p_j\rceil - 1)
     \\ \ge \F\left( \dottt \p_\ell \dottt,\; \lceil \p_{m+1}\rceil, \dottt, \lceil \p_{m+n}\rceil,\; 0,\dottt, 0,\;
         1 + \sum (\lceil \p_j\rceil - 1) - \sum (\lceil \p_j\rceil - \p_j) \right)
       - \sum (\lceil \p_j\rceil - 1)
     \\ = \F\left( \dottt \p_\ell \dottt,\; 1 + \sum\!{\raisebox{-2mm}{\footnotesize\it j}} (-1) - \sum (-\p_j) \right)
       + \sum \F(0,\dottt, 0,\; \lceil \p_j\rceil) - \sum (\lceil \p_j\rceil - 1)
     \\  = \F(\dottt \p_\ell \dottt,\; 1-n + \sum \p_j).
  \end{eqnarray*}
  The resulting equality implies\footnote{ Notice that $\sum \p_j = m+n - \sum \p_{\ell}$.}
  the property~\eqref{eq:(<)} upon recalling
  that the value of $\F$ does not change when removing the entries $\p_\ell \!=\! 1$ from $\P$, 
  so that $\P^\uparrow_\le$ becomes $\P^\uparrow_<$.
  Moreover, since (A) and (S) are among the defining properties of co-ENFs, $\F$ is additively separable by Lemma~\ref{le:separable}~(a).
\end{proof}

\begin{corollary}\label{cor:(<)}
   All functions $\F \in \Fs$ satisfy \eqref{eq:(<)}.
\end{corollary}

\bigskip

Next, we investigate the non-uniqueness of the generating function, which requires some groundwork to begin with. 

\begin{lemma}\label{le:mx+mz}  In real numbers, the following two statements are true:
  \dsskip
  \begin{align*}
    \/{i}\quad  & (\forall a) (\forall b>a) (\forall x\ne 1)   \quad(\exists m,n\in\Z) (\exists z\in [a,b])   \quad(mx + nz = m+n)
 \\ \/{ii}\quad & (\forall a>1) (\forall b>a) \;\;(\exists B)\;\;  \bigl( \forall x\in \bigl[ 0,{\textstyle\frac{1}{2}} \bigr] \bigr)
             \;\;(\exists m,n\in\Z^+) (\exists z\in [a,b]) \;\;\Bigl( mx + nz = m+n, \,\text{ and }\, \frac{n}{m}\le B \Bigr)
   \end{align*}
 \end{lemma}

\begin{proof}
  \/{i} Fit $\frac{m}{n}$ between $\frac{a-1}{1-x}$ and $\frac{b-1}{1-x}$ using the density of rationals in $\R$.
  Then, to get the equality, choose $z=1+ (1-x)\frac{m}{n}$.

\medskip\noindent
  \/{ii}  Choose $B= \frac{1}{a-1}$ and
  $n = \lceil \frac{1}{b-a} \rceil$, then
  \begin{equation*}
    \frac{1}{n} \le b-a \le \frac{b-a}{1-x} = \frac{b-1}{1-x} - \frac{a-1}{1-x}
  \end{equation*}
  since  $0 \le x < 1$.
  Hence, there exists $m$ to fit $\frac{m}{n}$ between $\frac{a-1}{1-x}$ and $\frac{b-1}{1-x}$.
  Then, again, choose  $z=1+ (1-x)\frac{m}{n}$ to get the equality.
  Moreover, 
  \begin{equation*}
    (a-1) \le \frac{a-1}{1-x} \le \frac{m}{n} \quad\text{and so}\quad \frac{n}{m} \le \frac{1}{a-1} = B.
  \end{equation*}
  This concludes the proof.
\end{proof}

\medskip

\begin{lemma}\label{le:generating}\setstretch{1.05}
  Let $\H$ be an additively separable function on $\Ps$
  and $\h,\h_1,\h_2$ its generating functions. Then,
\br7
  \/i $\h(\p) + (1-\p)K$ is also a generating function of $\H$ for every  number $K$,
\br7
  \/{ii} if $\h_1(\p)-\h_2(\p)$ is bounded
  on some interval $[a,b],\; 0\le a < b$,
    then $\h_1(\p)-\h_2(\p) = (1-\p)K_0$ for some  number $K_{0}\ $ and all $\p \in [0,\infty)$.
\end{lemma}
\begin{proof} 
  \/i This follows from $\sum_{i=1}^N (1-\p_i) = 0$ for all $\P=(...\p_i...)\in \Ps_N$.
\br7
  \/{ii}  
  Setting $\hh=\h_1 - \h_2$ gives the following equation:
\dsskip
  \begin{equation}\label{eq:H=0}
    \sum \hh(\p_i) = 0 \mathrm{\quad for\ all\ } N  \mathrm{\ and\ all\ } \P=(\p_1,\dottt,\p_N).  
  \end{equation}
  We will show that 
  \begin{equation}\label{eq:h=linear}
    \hh(\p) = (1 - \p)\,\hh(0)  \mathrm{\quad for\ all\ } \p\in [0,\infty).
  \end{equation}
  Case $\p=1$. We get $\hh(1)=0$ from \eqref{eq:H=0} for $N=1$, and then, \eqref{eq:h=linear} is satisfied for $\p=1$.
\br7
\setstretch{1.1}%
  Case $\p > 1$.
  Without loss of generality, we can assume that $1\notin [a,b]$.
  Moreover, if $0\le a < b < 1$, then $\hh$ is bounded also on $[2-b,\, 2-a]$
  because $\hh(2-\p) = - \hh(\p)$, which follows from \eqref{eq:H=0} when $N=2$.
  As a result, we can assume without loss of generality even that $1<a<b$.
  Under this assumption, we will first show that $\hh$, bounded on $[a,b]$, must be bounded also on $[0,\frac{1}{2}]$.
  According to Lemma \ref{le:mx+mz}(ii) we have
    \dsskip\[ (\forall a>1) (\forall b>a) \quad(\exists B)\quad \bigl( \forall x\in \bigl[ 0,{\textstyle\frac{1}{2}} \bigr] \bigr)
           \quad(\exists m,n\in\Z^+) (\exists z\in [a,b])  \quad\Bigl(mx + nz = m+n, \,\text{ and }\; \frac{n}{m}\le B \Bigr) \]
 Now, we choose $\P=(x,\dottt,x,z,\dottt,z) \in \Ps_{m+n}$
 in Equation \eqref{eq:H=0},
 where $x$ repeats $m$ times and $z$ repeats $n$ times.
 Then
   \[  m\hh(x)+n\hh(z)=0 \quad\text{and so}\quad |\hh(x)| =  \left| \frac{-n}{m}\hh(z) \right|  \le BA\quad , \quad\text{where} \]
 $A$ is a bound for $|\hh|$ on $[a,b]$.
 Thus we have $\hh$ bounded on $[0,\frac{1}{2}]$ by $BA$.

\medskip\noindent\setstretch{1.1}%
{\def\z{x}%
  Now, suppose by contradiction that there is $\p > 1$ such that $\hh(\p) \ne (1 - \p)\,\hh(0)$.
  Let $k$ be an integer large enough so that
  \dsskip
  \begin{equation}\label{cond:k}
     N - k - 2 = \lceil k\p \rceil - k - 2 \ge 0 \quad\text{and}\quad
     \frac{4BA}{k} < \bigl| \hh(c) - (1-c)\hh(0) \bigr|
  \end{equation}
  and set $N=\lceil k\p \rceil $ , $\z = \frac{1}{2}(N - k\p)$.
  Then, $k\p + 2 \z = N$ and $\z \in [0,\frac{1}{2}]$.
  Choosing $\P=(\p,\dottt,\p,0,\dottt,0,\z,\z)\in \Ps_N$ in \eqref{eq:H=0}, 
  where $\p$ repeats $k$ times, will produce 
  \dsskip\[ k\hh(\p) + (N - k - 2)\,\hh(0) + 2 \hh(\z) = 0 \]
  and the first condition in \eqref{cond:k} ensures that $\P$ will not have a negative number of zeroes.
  Given that $N = k\p + 2 \z$, we get
  \dsskip
  \begin{align*}
     \hh(\p) + \left( \p-1 + \frac{2 \z}{k} - \frac{2}{k} \right)\,\hh(0) + \frac{2 \hh(\z)}{k} = 0
  \end{align*}
  \begin{align}\label{eq:epsilon}
    \hh(\p) = (1 - \p)\,\hh(0) + \e,
  \end{align}
  where we set $ \e = - \frac{2}{k} \left( (\z-1)\,\hh(0) + \hh(\z) \right) $.
  Since $k\e$ is bounded by $4BA$  
  and from the second condition in \eqref{cond:k}, we have
  \begin{align*}
    0 < |\e| \le \frac{4BA}{k}
      < \bigl| \hh(\p) - (1 - \p)\,\hh(0) \bigr| ,
  \end{align*}
  which is a contradiction with \eqref{eq:epsilon}.
  Thus, the original assumption $\hh(\p) \ne (1 - \p)\,\hh(0)$ failed, and \eqref{eq:h=linear} holds for $\p>1$.
}

\medskip\noindent
  Case $\p < 1$.
  We can use the previous case for $2-\p > 1$ and get $\hh(2-\p) = (1 - (2-\p))\,\hh(0)$.
  This equation transforms into \eqref{eq:h=linear}
  since we already know that  $\hh(2-\p) = -\hh(\p)$,  
  and thus, \eqref{eq:h=linear} holds for $\p < 1$ as well.

\medskip\noindent
  We have shown that if $\hh$ satisfies \eqref{eq:H=0}, then it satisfies \eqref{eq:h=linear}.
  Finally, setting $K_0=\hh(0)$ completes the proof.
  %
\end{proof}

\begin{proposition}\label{prop:generating}
\br7
   Let $\F \in \Fs$. Then for each number $t$, there is a unique generating function $\f$ of $\F$  
   that is continuous and $\f(0)=t$.
\end{proposition}

\begin{proof}
 The existence of one continuous generating function, not necessarily satisfying $\f(0)=t$, follows from 
 Corollary \ref{cor:(<)} and Lemma \ref{le:separable}(b). Then part (i) of Lemma \ref{le:generating} implies that 
 there is at least one continuous generating function for an arbitrary value of $t=\f(0)$. 
 To show the uniqueness of such a function for every $t$, assume that there are two continuous generating 
 functions $\f_1$ and $\f_2$, such that $\f_1(0)=\f_2(0)$. Since $\f_1(\p) - \f_2(\p)$ is then bounded on any 
 finite interval due to continuity, we can use (ii) of Lemma \ref{le:generating} to infer that 
 $0 = \f_1(0)-\f_2(0) = K_0$.  Using (ii) of Lemma \ref{le:generating} again, we finally conclude 
 $\f_1(\p)=\f_2(\p)$, as claimed. 
  
\end{proof}


\subsection{Description and Structure of Co-ENFs}

\noindent
The separability results of the previous section give us access to the content and the structure of set $\Fs$, 
ultimately providing a key insight into the concept of effective (co-)numbers. We start with the following
proposition\footnote{The claim (i) of Proposition \ref{prop:convex} is likely to be known 
in the context of majorization but we haven't found a suitable reference.} 

\begin{proposition}\label{prop:convex}
\br7
  \/i Let $\H$ be a real additively separable function defined on $\Ps$.
    $\H$ is continuous (C) and monotone (M$^+$) if and only if it can be generated by a function $\h(\p)$ that is 
    continuous at $\p=0$ and convex.
\br7
  \/{ii} If $\F \in \Fs$ 
  then all its continuous generating functions $\f(\p)$ are convex. 
\end{proposition}
\begin{proof} %
  \/i \/{$\Leftarrow$} The convexity and continuity of $\h$ at $\p=0$ imply its continuity on $[0,\infty)$, which 
  guarantees the continuity (C) of $\H$. In the presence of additive separability, conditions entailed by (M$^+$) 
  take the form  
\dsskip
  \begin{align}
    \h(\p_i)+\h(\p_j) \le \h(\p_i-\e)+\h(\p_j+\e)\quad,\quad \p_i\le \p_j
  \end{align}
  To show that this also follows from the stated properties of $\h$, consider function $\hh$ which equals $\h$ 
   everywhere except on interval $[\p_i,\p_j]$, where it is replaced by a linear segment with boundary values 
   $\h(\p_i)$ and $\h(\p_j)$. Such function $\hh$ is still convex, which implies
  \begin{align*}
    \hh\left( \frac{(\p_i-\e)+(\p_j+\e)}{2} \right)  \le  \frac{\hh(\p_i-\e) + \hh(\p_j+\e)}{2}.
  \end{align*}
    Then, by linearity, the left-hand side is
   \[  \hh\left( \frac{\p_i+\p_j}{2} \right)
       = \frac{\hh(\p_i) + \hh(\p_j)}{2}  =  \frac{\h(\p_i) + \h(\p_j)}{2}. \]
    The inequality turns into
\begin{samepage}
\dsskip
  \begin{align*}
    \h(\p_i)+\h(\p_j) \le \hh(\p_i-\e) + \hh(\p_j+\e) = \h(\p_i-\e)+\h(\p_j+\e) 
  \end{align*}
  as needed. 
\end{samepage}
\br7
  \/{$\Rightarrow$} Consider the (M$^+$) condition  $\H(...\p_i...\p_j...) \le \H(...\p_i-\e...\p_j+\e...)$ for additively 
  separable $\H$. Setting $\p_i = \p_j = \p$ and, subsequently, $a=\p-\e,\; b=\p+\e$, we obtain in turn
  \begin{eqnarray*}
      \h(\p) + \h(\p)  &\le&  \h(\p-\e) + \h(\p+\e)
     \\ \h\left( \frac{a+b}{2} \right) &\le& \ds\frac{\h(a) + \h(b)}{2}.
  \end{eqnarray*} 
  Hence any $\h$, a generating function of $\H$, is midpoint convex on $[0,N]$ for all $N$. 
  It is well known that every such function is convex if it is continuous. 
  We thus select a continuous generating function $\h$, whose existence is guaranteed by Proposition 12. 
  Such a resulting $\h$ is then both continuous at $\p=0$ and convex.
  Note that $\h$ is an arbitrary continuous generating function, so all continuous generating functions $\h$ 
  are convex. This is needed in the proof of (ii) that follows.
\br7
   \/{ii} Proposition \ref{prop:monoseparable} implies the additive separability of $\F \in \Fs$ and the rest of 
   the demonstration is contained in the proof of (i)($\Rightarrow$) above.  
\end{proof}

\bigskip
We are now in a position to describe the set $\Fs$, specified by Definition \ref{def:fun}, explicitly. 

\begin{theorem}{\bf (Set of co-ENFs)}\label{th:eff}
\br5
  $\F \in \Fs$ if and only if it is generated by a convex and continuous function $\f$, which is zero on $[1,\infty)$ 
  and $\f(0)=1$. Such a generating function $\f$ of $\F$ is unique.
\end{theorem}


\begin{proof}
  \/{$\Rightarrow$} $\F \in \Fs$ is additively separable by Proposition \ref{prop:monoseparable}. Then, as 
  a consequence of additivity (A) and the boundary conditions (co-B2), we have 
  $\;\F(0,\dottt,0,N) = (N-1)\cdot \f(0) + \f(N) = N-1$, so that 
\dsskip
  \begin{align}\label{eq:B2c}
    \f(N) = (N-1)\cdot \bigl( 1 - \f(0) \bigr)
  \end{align}
  for all $N$. Given the continuity of $\F\in \Fs$, Proposition \ref{prop:generating} guarantees
  the existence of its unique continuous generating function with $\f(0)=1$.
  In conjunction with Eq.~\eqref{eq:B2c}, this implies that $\f(N)=0$ for all $N$.
  Furthermore, this continuous generating function is convex by (ii) of Proposition \ref{prop:convex},
  and consequently, it is zero on the entire $[1,\infty)$.
  This demonstrates the existence of unique $\f$ with all required properties.
\br7
  \/{$\Leftarrow$} For the opposite direction, let $\F[\P]=\sum \f(\p_i)$, where $\f$ is continuous, convex, 
  $\f(0)=1$, and $\f(\p)=0$ on $[1,\infty)$.
  Then (A), (C), (S), and (co-B2) follow immediately, while (M$^+$) is a consequence of 
  Proposition \ref{prop:convex}(i).
\end{proof}

\bigskip

Note that the unique choice of the continuous generating function for co-ENF is facilitated by a natural
choice $\f(0)=1$, expressing the fact that the object assigned zero probability should not contribute 
to the effective number total ($\cf(0)=0$). However, it is worth pointing out that, as shown below, the same 
unique choice of a generating function is selected by the requirement of the boundedness on the entire 
$[0,\infty)$.   

\begin{corollary}\label{cor:bounded}
  Let $\f$ be the generating function of $\F \in \Fs$, specified in Theorem \ref{th:eff}. Then
\br5
  \/{i} $0 \le \f(\p) \le 1$ for all $\p$
\br5
  \/{ii} $\f$ is the only generating function of $\F$ that is bounded on its whole domain $[0,\infty)$
\end{corollary}
\begin{proof}
   \/{i} This immediately follows from $\f(0) = 1$, $\f(\p)=0$ on $[1,\infty)$, and convexity. 
\br7
  \/{ii} The boundedness of $\f$ follows from (i). To demonstrate uniqueness, assume there is another
  bounded generating function $\f_1$ of $\F$. Thus, $\f_1 - \f$ satisfies the assumptions of 
  Lemma~\ref{le:generating}(ii), implying the existence of non-zero $K_0$ such that 
  $\f_1(\p) = \f(\p) + K_0 (1-\p)$ for $\p \in [0,\infty)$. However, this contradicts the boundedness
  of $\f_1$, which demonstrates the claimed uniqueness. 
\end{proof}

\bigskip
  Below, we will make use of the following obvious lemma and a simple corollary.
\begin{lemma}\label{le:less}
   Let $\H_1$ and $\H_2$ be real additively separable functions on $\Ps$.
   If there exist respective generating functions such that $\h_1(\p) \le \h_2(\p)$, for all $\p$,
   then $\H_1(\P) \le \H_2(\P)$, for all $\P \in \Ps$.
\end{lemma}

\begin{corollary}\label{cor:Bounded}
  If $\F \in \Fs$, then $0 \le \F[\P] \le N-1$, for all $\P \in \Ps$.
\end{corollary}
\begin{proof}
  Let $\f$ be the generating function specified in Theorem \ref{th:eff}.
  From (i) of Corollary \ref{cor:bounded}, we have $\sum_i 0 \le \sum \f(\p_i) \le \sum_i 1$, which 
  translates into  $0 \le \F[\P] \le N$ by Lemma \ref{le:less}. To put the second inequality into 
  the claimed form, note that there is always at least one $\p_j \ge 1$.
  For this $\p_j$, we have $\f(\p_j)=0$ by Theorem \ref{th:eff}.
  This lowers the upper bound for 
  $\F[\P]$ by unity and proves the second inequality.
\end{proof}

\bigskip
Using the above preparation, we will now demonstrate several structural properties of $\Fs$. 

\begin{theorem}{\bf (Maximality)}\label{th:max}\setstretch{1.45}
\br1
  If $\F \in \Fs$, then the following holds for all $\P = (...,\p_i,...) \in \Ps$,
  \\ \/i $\F_{(0)}[\P] \,=\, \Fo[\P] \,\le\, \F[\P] \,\le\, \F_{\star}[\P] \,=\, \F_{(1)}[\P]$
  \\ \/{ii} $\Fo[\P] \,=\, \F[\P] \,=\, \F_\star[\P]  \quad\Leftrightarrow\quad  \p_i\notin (0,1),\; i=1,2,...,N$
  \\ \/{iii} $\beta_0 \,=\, \Fo[\P] \;<\; \F_\star[\P] \,=\, \beta_1  \;\,\Rightarrow\;\,
   \{ \F[\P]: \F\in \Fs \}  =  [\beta_0,\beta_1]$  
\end{theorem}
\begin{proof}\setstretch{1.05} %
  \/i To show both inequalities, let $\f$ be the continuous generating function of $\F$ guaranteed by 
   Theorem \ref{th:eff}. We will show that $\f$ satisfies $\fo(x) \le \f(x) \le \f_{\star}(x)$ on $[0,\infty)$, and
   then, Lemma \ref{le:less} will complete the proof of this part. The first inequality $\fo(x) \le \f(x)$ follows 
   directly from Theorem \ref{th:eff}, Corollary \ref{cor:bounded}(i), and the definition of $\fo$.
   The second inequality holds as the equality on $[1,\infty)$ by Theorem \ref{th:eff} and the definition 
   of $\f_{\star}$. To show the second inequality on $[0,1)$, note that the graphs of both $\f$ and $\f_{\star}$ 
   pass through the points $(0,1)$ and $(1,0)$, $\f$ is convex by Proposition \ref{prop:convex}(ii), and 
   $\f_{\star}$ is linear between those points. Therefore, $\fo(x) \le \f(x) \le \f_{\star}(x)$ on $[0,\infty)$ as promised.
\br7
   \/{ii} \/{$\Leftarrow$} If $\p_i\notin (0,1)$ for $i=1,2,...,N$, then $\Fo[\P] = \F_\star[\P]$,
     and the equality for $\F[\P]$ follows from (i).
\br7
  \/{$\Rightarrow$} Let $\Fo[\P] = \F_\star[\P]$. Then,
  \[ \sum_{\p_i=0} 1  =  \sum_{\p_i=0} 1 + \sum_{\p_i \in (0,1)} (1 - \p_i). \]
  Hence, $\; \sum_{\p_i \in (0,1)} (1 - \p_i) = 0$, and so, $\p_i\notin (0,1)$ for $i=1,2,...,N$.
\br7
   \/{iii} From $\Fo[\P] < \F_\star[\P]$ we have the existence of at least one $\p_i \in (0,1)$ by (ii).
   Then, for a fixed $\P$, define
  \[ g(\alpha) = \sum_{\p_i=0} 1 + \sum_{\p_i \in (0,1)} (1 - \p_i^\alpha) = \F_{(\alpha)}[\P]. \]
   The function $g$ is continuous, increasing, and maps interval $[0,1]$ onto interval $[\beta_0,\beta_1]$.
   Then
  \[ \{ \F[\P]: \F\in \Fs \}  \supseteq  \bigl\{ \F_{(\alpha)}[\P]: \alpha \in [0,1] \bigr\}
         =  \bigl[ \Fo[\P],\; \F_\star[\P] \bigr] = [\beta_0,\beta_1]. \]
  The opposite inclusion follows from (i).
\end{proof}


\section{Counting The General Quantum Identities}
\label{sec:count_idents}

With the effective number theory in place, we now return to the topic that motivated its 
construction, namely the quantum identity problem. In particular, we will make explicit 
some of the straightforward, but useful and relevant, generalizations of [Q]. This serves, 
in part, as a stepping stone toward the most generic application of effective numbers 
in quantum theory, namely as a tool to characterize quantum states 
(Sec.~\ref{sec:str_states}).

A conceptually important extension of [Q] and [A] is made possible by the additive separability 
of ENFs. Indeed, instead of an orthonormal basis, consider any collection $\{ \, \ketj \}$ 
of $n$ orthonormal states from an $\nrN$-dimensional Hilbert space ($1 \!\le\! n \!\le\! \nrN$). 
How many states from $\{ \, \ketj \}$ is a system described by $\ketpsi$ effectively in? 
Let $\efN \in \Nmaps$ be an ENF and $\cf$ the counting function uniquely associated with
it by virtue of Theorem$\,$1. We define
\begin{equation}
     \efN\bigl[\, \ketpsi , \{ \, \ketj \} \,\bigr] \,\equiv\, 
     \sum_{j=1}^n \cf(\w_j)  \quad\; , \quad\;
     \w_j = \nrN  \mid \!\! \langle \,j \!\mid \! \psi \,\rangle \!\! \mid^2 \quad
     \label{eq:4_10}                            
\end{equation}
for each $\ketpsi$ and $\{ \, \ketj \}$. This assignment is meaningful in the following 
sense. Given a fixed $\{ \, \ketj \}$, let $\{ \, \keti \}$ be its arbitrary completion 
into a basis of the Hilbert space. Then, owing to the additive separability of ENFs, 
\begin{equation}
     \efN \bigl[\, \ketpsi , \{ \, \keti \} \,\bigr]  \;=\;
     \efN \bigl[\, \ketpsi , \{ \, \ketj \} \,\bigr]  \;+\;  
     \efN \bigl[\, \ketpsi , \{ \, \keti \} \setminus \{ \, \ketj \} \,\bigr]
     \label{eq:4_20}
\end{equation}
where ``$\setminus$'' denotes the set subtraction. In other words, the contribution of 
$\{ \, \ketj \}$ to $\efN[\, \ketpsi , \{ \, \keti \} \,]$, defined by \eqref{eq:4_10}, is 
independent of the completion $\{ \, \keti \}$, and is thus uniquely associated with 
$\{ \, \ketj \}$ for fixed $\efN \in \Nmaps$. It is the effective number of states from 
$\{ \, \ketj \}$ contained in $\ketpsi$ according to $\efN$. Moreover, it is 
straightforward to check that $\efNm[\, \ketpsi , \{ \, \ketj \} \,]$ minimizes 
the effective number so assigned, and we have

\medskip

\noindent
    {\bf [A$^\prime$]} {\em Let $\ketpsi$ be a state vector from an $\nrN$-dimensional 
    Hilbert space 
    and $\{ \, \ketj \}  \equiv  \{ \, \ketj \mid\, j=1,2,\ldots,n \le \nrN \,\}$ the set of 
    $n$ orthonormal states in this space. The physical system described by $\ketpsi$ is 
    effectively in $\efNm [\, \ketpsi , \{ \, \ketj \} \,]$ states from 
    $\{ \, \ketj \}$, specified by~\eqref{eq:4_10} with $\cf=\cfu$.}

\medskip

\noindent
A few simple points regarding the above are worth emphasizing.

\medskip

\noindent
(i) The extension \eqref{eq:4_10} and the ensuing generalization of [A] to [A$^\prime$] 
arise because the abundance of quantum identities is determined ``locally'', namely 
without reference to basis elements orthogonal to the subspace in question. Apart from 
being natural for a measure-like characteristic, this feature has practical consequences 
in many-body applications where the dimension of the Hilbert space grows exponentially 
with the size of the system. Indeed, the above avoids such complexity in certain calculations,
thus providing a computational benefit.

\medskip

\noindent
(ii) None of the above applies to the abundance of quantum identities determined by 
the participation number $\pN$ of Eq.~\eqref{eq:045}, since this value does not split 
into contributions from orthogonal subspaces generated by the partitioned basis. 
This is of course due to the lack of additivity and hence of additive separability. 

\medskip

\noindent
(iii) The above considerations are clearly not limited to counting quantum identities. 
In a generic situation, the inquiry is concerned with the contribution to the effective 
number from a subset of weighted objects. To formalize such assignments directly in 
the effective number theory, one simply extends $\efN \in \Nmaps$ from domain $\setW$ 
of counting vectors to the domain of general weights
\begin{equation}
    \setWl=\union_n \setWl_n    \qquad , \qquad
    \setWl_n  \,\equiv\, 
    \bigl\{ \, \Wl=(\wl_1,\wl_2, \ldots, \wl_n) \;\mid\;  \wl_j \ge 0 \; \bigr\} \quad  
    \label{eq:4_30}
\end{equation}
by virtue of its counting function $\cf$, namely $\efN[\Wl]=\sum_j \cf(\wl_j)$. 

\smallskip

While the effective number $\efN[\, \ketpsi , \{ \, \keti \} \,]$ specifies how many 
identities from basis $\{ \, \keti \}$ the state 
$\ketpsi$ can effectively take, it is frequently useful to inquire about 
a coarse-grained version of such an effective count. To formalize the corresponding 
generalization, consider 
the decomposition of Hilbert space $\Hsp$ into $M$ mutually orthogonal subspaces 
$\{ \Hsp_m \}$
\begin{equation}
     \Hsp = \Hsp_1  \oplus \Hsp_2 \oplus \ldots \oplus \Hsp_M    \qquad , \qquad
     \{ \Hsp_m \}  \equiv  \{ \, \Hsp_m \mid\, m=1,2,\ldots,M \,\}
     \label{eq:4_40}                         
\end{equation}
Subspaces $\Hsp_m$ are implicitly treated as equivalent entities.\footnote{What is
considered ``equivalent'' is dictated by the physics involved, rather than the concept 
of effective number itself. For example, it is possible to encounter a situation 
where the dimensions of ``usefully equivalent'' subspaces are not the same.} 
In how many subspaces from $\{ \Hsp_m \}$ is the system described by $\ketpsi$ 
effectively?

Given a state of the system, the rules of quantum mechanics assign a probability to each 
subspace of the associated Hilbert space. The effective number theory therefore 
provides an immediate answer to the above question. More specifically, let $\ketchim$ be 
the (unnormalized) projection of $\ketpsi$ onto $\Hsp_m$. Then the probability vector 
is $P=(p_1,p_2,\ldots,p_M)$, where $p_m = {\langle \,\chi_m \!\mid\! \chi_m \,\rangle}$, 
and the corresponding counting vector is $\W=MP$. Hence, given an ENF, we assign 
$\efN[\, \ketpsi , \{ \, \Hsp_m \} \,] \equiv \efN[\W]$, which leads to the following
generalization [A$_g$] of [A].

\medskip

   \noindent
   {\bf [A$_g$]} {$\,$ \em Let $\W$ be the counting vector assigned by quantum mechanics 
   to state $\ketpsi$ and the orthogonal decomposition $\{ \Hsp_m \}$ of the Hilbert 
   space. Then the system described by $\ketpsi$ is effectively contained in 
   $\efNm[\, \ketpsi , \{ \, \Hsp_m \} \,] = \efNm[\W]$ subspaces from $\{ \, \Hsp_m \}$.}  
 
\medskip


Applying the logic identical to the one producing [A$^\prime$], it is straightforward 
to generalize [A$_g$] into [A$^\prime_g$] for counting the identities from arbitrary sets 
(not necessarily full decompositions) of mutually orthogonal subspaces.


\section{Application: The Structure of Quantum States}
\label{sec:str_states}

While our discussion was carried out in the context of the quantum identity problem, 
the constructed effective number framework offers a much larger scope of 
uses. Apart from physics and mathematics, these also appear in the areas of 
applied science simply due to the very basic role of the measure and probability in 
quantitative analysis. Examples of such applications will be discussed in 
the follow-up works. Here we describe a broader outlook on the utility of 
effective numbers in quantum theory.

With the quantum state encoding all ``options'' for the system, the physical 
content of $\ketpsi$ closely relates to all probability vectors $P$ induced
in this manner. 
Retracing the steps leading to [A], it is then meaningful to associate 
the effective number $\efNm$ with any complete set of mutually exclusive 
possibilities. We propose the collection of all such reductions 
\begin{equation}
    \ketpsi  \quad \longrightarrow \quad  P 
    \quad \longrightarrow \quad   \efNm
    \label{eq:4_45}
\end{equation}
as a systematic and physically relevant characterization of a quantum state. 
The general quantum identities of Sec.~\ref{sec:count_idents}, which can 
also be thought of as measurement outcomes, are the prime examples 
of objects/possibilities in question. However, any meaningful 
$\ketpsi \rightarrow P$ can be considered. The resulting variety offers a wide 
range of options for targeted insight into the structure of $\ketpsi$. 
We wish to highlight a few points in this regard. 

\smallskip

\noindent (i)
Each effective number characteristic can be refined by totals assigned 
to subsets of the associated sample space or coarse-grained by counting 
its partitions. Note that for quantum identities in [Q], this corresponds 
to refining by totals involving parts of the basis and coarse-graining 
by totals involving orthogonal decompositions.

\smallskip

\noindent (ii)
The proposed approach is universal with respect to the nature of 
the quantum system being described. Indeed, $\ketpsi$ may be as simple 
as a state of the harmonic oscillator, but as complex as a many-body state of 
quantum spins or the vacuum of non-Abelian gauge theory.

\smallskip

\noindent (iii)
Our present focus on the finite discrete case is not restrictive either. Even 
in situations involving the space-time continuum, the intermediate steps of 
ultraviolet (lattice) and infrared (volume) regularizations yield such 
descriptions. Removing the cutoffs then involves a suitable 
$\nrN \!\to\! \infty$ limit. 
For this purpose, it is more convenient to work with the effective fraction
rather than the effective number. In terms of probabilities, it is defined as
 \begin{equation}
    \efF[P] \,\equiv\, \frac{1}{\nrN} \; \efN[ \nrN P ]  
    \quad , \quad   \efN \in \Nmaps  
    \quad,\quad P \in \setP_\nrN    \quad
    \label{eq:105}                        
 \end{equation}
Function $\efNm$ produces the minimal effective fraction, namely
\begin{equation}
     \efFm[P] = \sum_{i=1}^\nrN \ffu(p_i,\nrN)    \quad,\quad
     \ffu(p,\nrN)  \equiv  \min\, \{\, p, 1/\nrN \,\}    \;
     \label{eq:115}         
\end{equation}
Restricting ourselves to quantum identities of [Q], the cutoff removal schematically 
proceeds as follows. Assume that $\ketpsi$ is the target state 
(from the infinite-dimensional Hilbert space) 
to be assigned an effective fraction in basis $\{ \,\keti \}$. 
Let ${\mid \! \psi^{\scrk} \,\rangle}$ represent $\ketpsi$ at the $k$-th step 
of the regularization process, involving the state space of increasing dimension 
$\nrN_k$. If the latter is spanned by $\{ \,\mid \! i^{\scrk} \,\rangle \}$ 
targeting $\{ \,\keti \}$ then~\footnote{The meaning of
$\{ \,\mid \! i^{\scrk} \,\rangle \} \equiv 
\{ \,\mid \! i^{\scrk} \,\rangle  \mid i=1,2,\ldots,\nrN_k \}$
targeting $\{ \,\mid \! i \,\rangle \}$ depends on the context but is usually 
clear on physics grounds.} 
\begin{equation}
     \efFm \bigl[ \,\mid \! \psi \,\rangle, \, \{ \,\mid \! i \,\rangle \} \, \bigr]  
     \;\equiv\; 
     \lim_{k \to \infty} \efFm 
     \bigl[ \,\mid \! \psi^{\scrk} \,\rangle, \, \{ \,\mid \! i^{\scrk} \,\rangle \} \, \bigr] 
     \;=\; 
     \lim_{k \to \infty} \efFm [P_k] 
     \label{eq:117}
\end{equation}
where $P_k \!=\! (p_1^\scrk,\ldots,p_{\nrN_k}^\scrk)$ , 
$p_i^\scrk \!=\! \,{\mid \! \langle \,i^{\scrk} \!\mid \! \psi^{\scrk} \,\rangle \! \mid^2}$. 
Note that  $\efFm$ reflects the ``localization'' of $\ketpsi$ in $\{ \,\keti \}$.

\smallskip

\noindent (iv)
An example of an application where the assignment \eqref{eq:4_45} is carried out without 
direct reference to the underlying Hilbert space is provided by the problem of the vacuum 
structure in quantum chromodynamics (QCD). This is often studied in the Euclidean path 
integral formalism, with the regularized vacuum represented by the statistical ensemble of 
lattice gauge configurations $U\!=\! \{U_{x,\mu} \}$. 
Given a composite field $O\!=\!O(x,U)$ and the induced space-time probability 
distribution $P(x,U) \propto |O(x,U)|$, the effective number framework can be used 
to determine the effective fraction of space-time occupied by $O$ for each $U$.
The corresponding quantum averages are of vital interest in this context.
Yet more indirect vacuum characteristics of such a type reflect the space-time properties 
of Dirac eigenmodes, whose main utility is to probe the features of quark dynamics. 


\section{Concluding Remarks}
\label{sec:conclude}

When a quantum system is (strongly) probed, it emerges in one of many possible "identities". 
This is among the key features of quantum behavior. Indeed, it underlies 
the notion of quantum uncertainty and is closely connected to a fruitful concept 
of localization. A well-founded prescription for the corresponding abundances is thus 
desirable. As a contemporary example, one may use it in the analysis of a quantum algorithm 
that produces the state $\ketpsi_o$ as an output of a quantum run and follows up 
with a measurement involving a basis $\{ \,\keti \}$. The effective number 
of distinct collapsed states $\keti$ obtained upon repetition of these steps is relevant 
for the assessment of the algorithm's efficiency. 

In this work, we showed that requiring the desired characteristics to be measure-like 
(additive) is fruitful in identifying a meaningful quantifier. In particular,
it results in the theoretical structure (effective number theory) revealing that 
a consistent assignment of totals to collections of objects with probability weights 
requires the existence of an inherent (minimal) amount $\efNm$. 
The appearance of such a qualitative feature in basic measure considerations suggests 
the utility of the constructed framework already in contexts much less abstract than 
quantum mechanics. For example, $\efNm$ can be viewed as an extension 
of the ordinary counting measure into what can be referred to as a {\em diversity measure} 
with its wide range of contemporary applications 
(social sciences, ecosystems; see e.g.~\cite{Lei12A,Lei16A}). 
Other viewpoints can cast it as a {\em choice measure}, facilitating a probabilistic 
notion of effective choices, or as a {\em support measure}, conveying the effective size 
of a function support (effective domain). Given this universality, $\efNm$ may find uses in 
multiple areas of quantitative science.


Finally, we wish to point out that the existence of a minimal amount is rooted in the simultaneous 
requirement of both monotonicity (Schur concavity) and additivity for ENFs. This combination is 
rather unusual from the mathematics standpoint. Indeed, while monotonicity is important for 
the theory of majorization~\cite{Mar11A}, it has no role in the standard formalization 
of the measure. Conversely, additivity is crucial for the latter but not native to the 
former. In fact, relaxing either \eqref{eq:mon-} or \eqref{eq:add} in Definition~0 leaves 
the respective effective pseudo-number assignments too arbitrary. However, their combination, 
which is necessary on conceptual grounds, leads to $\efNm$ and the associated insights.


\begin{acknowledgments}
{\bf Acknowledgments. }
I.H. acknowledges the support by  the Department of Anesthesiology at
the University of Kentucky.
\end{acknowledgments}



\end{document}